\documentclass[conference,a4paper]{IEEEtran}
\IEEEoverridecommandlockouts

\usepackage[utf8]{inputenc} 
\usepackage[T1]{fontenc}
\usepackage{url}              
\usepackage{cite,setspace}             

\usepackage[cmex10]{amsmath}  
\usepackage{mleftright}       
\mleftright                   

\usepackage{graphicx}         
\usepackage{booktabs}         

\usepackage{verbatim}
\usepackage{enumerate}
\usepackage[mathcal]{euscript}
\usepackage{amssymb,amsfonts,amsthm}
\usepackage{marginnote} 

\usepackage{microtype}

\usepackage{enumitem}

\usepackage[usenames]{color}
\definecolor{plum}  {rgb}{.4,0,.4}
\definecolor{BrickRed} {rgb}{0.6,0,0}

\usepackage{commath}
\usepackage[normalem]{ulem}

\newcommand\restr[2]{{
  \left.\kern-\nulldelimiterspace 
  #1 
  \littletaller 
  \right|_{#2} 
  }}

\newcommand{\littletaller}{\mathchoice{\vphantom{\big|}}{}{}{}}

\def\ddefloop#1{\ifx\ddefloop#1\else\ddef{#1}\expandafter\ddefloop\fi}

\def\ddef#1{\expandafter\def\csname c#1\endcsname{\ensuremath{\mathcal{#1}}}}
\ddefloop ABCDEFGHIJKLMNOPQRSTUVWXYZ\ddefloop

\def\ddef#1{\expandafter\def\csname s#1\endcsname{\ensuremath{\mathsf{#1}}}}
\ddefloop ABCDEFGHIJKLMNOPQRSTUVWXYZ\ddefloop

\def\E{\mathbf{E}}
\def\Pr{\mathbf{P}}

\def\Reals{\mathbb{R}}
\def\Naturals{\mathbb{N}}

\def\deq{:=}

\def\bd#1{\boldsymbol{#1}}

\def\1{{\mathbf 1}}

\def\eps{\varepsilon}

\newtheorem{theorem}{Theorem}

\newtheorem{proposition}{Proposition}
\newtheorem{corollary}{Corollary}

\begin{document}

\title{Majorizing Measures, Codes, and Information} 

\author{Yifeng Chu and Maxim Raginsky
\thanks{The authors are with the Department of Electrical and Computer Engineering and the Coordinated Science Laboratory, University of Illinois, Urbana, IL 61801, USA. Emails: ychu26@illinois.edu, maxim@illinois.edu.}
\thanks{This work was supported by
the Illinois Institute for Data Science and Dynamical Systems (iDS${}^2$), an
NSF HDR TRIPODS institute, under award CCF-1934986.}}

\maketitle

\begin{abstract}
The majorizing measure theorem of Fernique and Talagrand is a fundamental result in the theory of random processes. It relates the boundedness of random processes indexed by elements of a metric space to complexity measures arising from certain multiscale combinatorial structures, such as packing and covering trees. This paper builds on the ideas first outlined in a little-noticed preprint of Andreas Maurer to present an information-theoretic perspective on the majorizing measure theorem, according to which the boundedness of random processes is phrased in terms of the existence of efficient variable-length codes for the elements of the indexing metric space. 
\end{abstract}


	\section{Introduction}
	
	Let $(T,d)$ be a compact metric space, and let $(X_t)_{t \in T}$ be a zero-mean random process indexed by the elements of $T$ and satisfying the increment condition
	\begin{align}\label{eq:increment}
		\Pr \Big[ |X_s - X_t| \ge u d(s,t) \Big] \le 2e^{-u^2}, \quad s,t \in T,\, u > 0.
	\end{align}
The question of boundedness of the expected supremum $\E[\sup_{t \in T}X_t]$ arises in a variety of applications, such as high-dimensional probability, convex geometry, functional analysis, theoretical computer science, mathematical statistics, and learning theory \cite{Talagrand_2014}. A fundamental result of Fernique \cite{Fernique_1975} and Talagrand \cite{Talagrand_1987} states that, for any fixed but arbitrary $t_0 \in T$,
\begin{align}\label{eq:mm_upper_bound}
	\E\Big[\sup_{t \in T}|X_t - X_{t_0}|\Big] \lesssim \inf_{\mu \in \cP(T)}\sup_{t \in T} I_\mu(t),
\end{align}
where $\cP(T)$ is the set of all Borel probability measures on $T$, 
\begin{align}\label{eq:FT}
	I_\mu(t) \deq \int^{{\rm diam}(T)}_0 \sqrt{\log \frac{1}{\mu(B(t,\eps))}}\dif \eps
\end{align}
is the \textit{Fernique--Talagrand functional}, and $a \lesssim b$ indicates that $a \le Cb$ for some absolute constant $C > 0$. Here, ${\rm diam}(T) \deq \sup \{d(s,t) : s,t \in T\}$ is the (finite) diameter of $T$, $B(t,\eps)$ denotes the closed ball of radius $\eps$ centered at $t$, and the logarithm is taken to base $2$. Moreover, if $(X_t)_{t \in T}$ is a Gaussian process and $T$ is endowed with the canonical metric
\begin{align}\label{eq:can_metric}
	d(s,t) \deq \sqrt{\E[(X_s-X_t)^2]},
\end{align}
then the bound \eqref{eq:mm_upper_bound} can be reversed \cite{Talagrand_1987}:
\begin{align}\label{eq:mm_lower_bound}
	\E\Big[\sup_{t \in T}|X_t - X_{t_0}|\Big] \gtrsim \inf_{\mu \in \cP(T)}\sup_{t \in T} I_\mu(t).
\end{align}
The inequality \eqref{eq:mm_upper_bound} says that the expected supremum of $(X_t)$ is bounded provided there exists at least one $\mu \in \cP(T)$, such that $\sup_{t \in T}I_\mu(t) < \infty$ (such probability measures are called \textit{majorizing measures}); moreover, for Gaussian processes, the existence of a majorizing measure is equivalent to boundedness of expected suprema. 

While the Fernique--Talagrand functional provides sharp bounds on the expected supremum of $(X_t)$, the machinery of majorizing measures was eventually replaced by an equivalent combinatorial framework involving various multiscale discrete structures, such as covering and packing trees (see, e.g., \cite{Guedon_2003} or \cite[Ch.~6]{Talagrand_2014}). In a little-noticed preprint \cite{Maurer_majorizing}, A.~Maurer has provided an alternative information-theoretic perspective on majorizing measures starting from the observation that the integrand in  \eqref{eq:FT} has an information-theoretic flavor--indeed, $\log \frac{1}{\mu(B(t,\eps))}$ is (roughly) proportional to the number of bits one would use to localize a point $t \in T$ up to a ball of radius $\eps$ when given a `prior' probability measure $\mu \in \cP(T)$. Building on this intuition, Maurer used the well-known correspondence between probability measures and uniquely decodable codes on a discrete (finite or countably infinite) alphabet \cite[Sec.~5.5]{Cover_Thomas} to show that the Fernique--Talagrand functional can be bounded both from above and from below in terms of information-theoretic quantities pertaining to sequences of variable-length lossy codes on $T$ that allow for arbitrarily accurate reconstruction of the points $t \in T$ using sufficiently long codewords, with the metric $d$ playing the role of the fidelity criterion \cite{Kontoyiannis_2002}.

In this paper, we develop this perspective further by giving a more direct construction of multiscale variable-length codes that allows us to seamlessly recover a number of existing bounds on expected suprema. Moreover, we derive an information-theoretic version of the lower bound \eqref{eq:mm_lower_bound} for Gaussian processses that matches our upper bounds up to a multiplicative constant. In terms of applications, we expect these results to be of use in analyzing the statistical performance of machine learning algorithms through the lens of data compression and the Minimum Description Length principle \cite{Massart_concentration,Audibert_Bousquet}.

	\section{Preliminaries}
	
	\subsection{Variable-length codes on metric spaces}
	\label{sec:VLC}
	
Let a compact metric space $(T,d)$ be given, with ${\rm diam}(T) = 1$. A \textit{variable-length code} (VLC) on $T$ is a pair $(\pi,f)$, where $\pi : T \to T$ is a mapping whose image $\pi(T)$ is a discrete  subset of $T$ and $f : \pi(T) \to \{0,1\}^*$ is a uniquely decodable representation of the elements of $\pi(T)$ by finite-length binary strings. In particular, it satifsies the Kraft--McMillan inequality
\begin{align}\label{eq:Kraft}
	\sum_{s \in \pi(T)} 2^{-\ell (f(s))} \le 1,
\end{align}
where $\ell(\cdot)$ denotes the length of a binary string. We say that a VLC $(\pi,f)$ \textit{operates at resolution $\rho \ge 0$} if $d(t,\pi(t)) \le \rho$ for all $t \in T$ (if $\rho = 0$, then $T$ is necessarily discrete, and $(\pi,f)$ is lossless). A natural procedure for constructing such VLC's is to fix a $\rho$-partition $\cA$ of $T$, i.e., a family of disjoint subsets $A \subseteq T$ of diameter at most $\rho$ whose union is all of $T$; if $\rho > 0$, then it is always possible to find a finite partition by compactness of $T$. Then, picking a point $s_A$ in each $A \in \cA$, we take $\pi(t) \deq s_{A(t)}$, where $A(t)$ is the unique element of $\cA$ that contains $t$. The mapping $\pi$ constructed in this way is evidently idempotent, i.e., $\pi \circ \pi = \pi$. Finally, with $\pi(T) = \{ s_A : A \in \cA\}$, we can choose any mapping $f : \pi(T) \to \{0,1\}^*$ that satisfies \eqref{eq:Kraft}. For instance, we can pick any $\mu \in \cP(T)$ such that $\mu(A) > 0$ for all $A \in \cA$ and let $f$ be the Shannon code with codelengths $\ell(f(s_A)) = \lceil \log \frac{1}{\mu(A)}\rceil, A \in \cA$. 

We will be working with \textit{sequences} of VLC's that operate at multiple scales. We will say that $\cC = \{(\pi_k,f_k)\}_{k \in \Naturals}$ is \textit{admissible} if the following conditions are satisfied:
\begin{itemize}
	\item[(i)] for each $k \in \Naturals$, $\pi_k$ is idempotent and there exists a mapping $\varphi_{k|k+1} : T \to T$, such that $\pi_k = \varphi_{k|k+1} \circ \pi_{k+1}$;
	\item[(ii)] $(\pi_k,f_k)$ operates at resolution $\rho_k$, and $\rho_k \to 0$ as $k \to \infty$;
	\item[(iii)] for each $t$ and each $k$, $\ell(f_k \circ \pi_k(t)) \le \ell(f_{k+1} \circ \pi_{k+1}(t))$.
\end{itemize}
We also define $(\pi_0,f_0)$ with $\pi_0(T) = \{t_0\}$ and $\ell(f_0(t_0)) = 0$ for a fixed but arbitrary point $t_0 \in T$. Condition (i) says that, for each $t$, the encoding $\pi_k(t)$ can be determined from any $\pi_m(t)$, $m > k$, without knowledge of $t$. We can take $\varphi_{k|k+1} = \pi_k$ without loss of generality since $\pi_k = \varphi_{k|k+1} \circ \pi_{k+1} = \varphi_{k|k+1} \circ \pi_{k+1} \circ \pi_{k+1} = \pi_k \circ \pi_{k+1}$. This implies that the partition $\cA_{k+1}$ induced by $\pi_{k+1}$ is a refinement of $\cA_k$, and for each $t \in T$ the sets $A_k(t)$ (of diameter $\rho_k$) are nested, i.e., $T \equiv A_0(t) \supseteq A_1(t) \supseteq A_2(t) \supseteq \dots$. Conditions (i) and (ii) imply that $\rho_k \searrow 0$  as $k \to \infty$. Finally, condition (iii) constrains the binary descriptions of each $t$ to become more `informative' as $k$ increases.

A natural procedure for constructing an admissible sequence of VLC's is to start with an increasing sequence $(\cA_k)_{k \ge 1}$ of partitions of $T$ such that each $A \in \cA_k$ has diameter at most $\rho_k$, with $\rho_k \searrow 0$ as $k \to \infty$. The mappings $\pi_k$ are constructed exactly in the same manner as for a single partition, so that conditions (i)--(ii) will be satisfied automatically. To construct the binary encodings $f_k$ satisfying condition (iii), choose some $\mu_1 \in \cP(T)$ such that $\mu_1(A) > 0$ for all $A \in \cA_1$; then, for each $k \ge 1$ and each $B \in \cA_k$, choose a probability measure $\nu_{k+1}(\cdot|B) \in \cP(T)$ such that $\nu_{k+1}(A|B) > 0$ for each $A \in \cA_{k+1}(B) \deq \{ A \in \cA_{k+1} : A \subseteq B \}$ and $0$ otherwise. This defines a sequence of probability measures $(\mu_k)_{k \ge 1}$ by
\begin{align}\label{eq:mu_k}
	\mu_{k+1}(\cdot) &= \sum_{B \in \cA_k}\nu_{k+1}(\cdot|B)\mu_k(B),
\end{align}
and we can choose each $f_k$ so that (ignoring the rounding issues) for each $t \in T$ and each $k$ we have
\begin{align*}
	&\ell(f_{k+1} \circ \pi_{k+1}(t)) \approx \log \frac{1}{\mu_{k+1}(A_{k+1}(t))} \\
	&\quad = \log \frac{1}{\mu_k(A_k(t))} + \log \frac{1}{\nu_{k+1}(A_{k+1}(t)|A_k(t))} \\
	&\quad \ge \log \frac{1}{\mu_k(A_k(t))} \\
	&\quad \approx \ell(f_k \circ \pi_k(t)),
\end{align*}
where in the second line we have used the fact that $A_{k+1}(t) \in \cA_{k+1}(A_k(t))$, and therefore $\mu_{k+1}(A_{k+1}(t)) = \nu_{k+1}(A_{k+1}(t)|A_k(t))\mu_k(A_k(t))$.

\subsection{Some facts about the Fernique--Talagrand functional}

The Fernique--Talagrand functional \eqref{eq:FT} can be expressed in terms of increasing sequences of partitions. To that end, following Bednorz \cite{Bednorz2015}, we define the following functional on pairs of probability measures $\mu,\nu \in \cP(T)$:
\begin{align*}
	M(\mu,\nu) \deq \int_T \int^1_0 \sqrt{\log\frac{1}{\mu(B(t,\eps))}}\dif \eps\, \nu(\dif t).
\end{align*}
Then $I_\mu(t) = M(\mu,\delta_t)$, where $\delta_t$ is the Dirac measure centered at $t$, and the bounds \eqref{eq:mm_upper_bound} and \eqref{eq:mm_lower_bound} can be expressed as
\begin{align*}
	\E\Big[\sup_{t \in T}X_t\Big] \lesssim \inf_\mu \sup_t M(\mu,\delta_t) =  \inf_\mu \sup_\nu M(\mu,\nu)
\end{align*}
and (for Gaussian processes on $T$ with the induced metric)
\begin{align*}
	\E\Big[\sup_{t \in T}X_t\Big] \gtrsim \inf_\mu \sup_t M(\mu,\delta_t) =  \inf_\mu \sup_\nu M(\mu,\nu).
\end{align*}
We then have the following upper bounds on $I_\mu(t)$ and $M(\mu,\nu)$ \cite{Bednorz2015}: Let $\cA = (\cA_k)_{k \ge 1}$ be an increasing sequence of partitions of $T$, such that ${\rm diam}(A) \le r^{-k}$ for all $A \in \cA_k$ and all $k$ for some $r > 1$. Then, setting $\cA_0 = \{T\}$, we have
\begin{align*}
	I_\mu(t) \le \sum_{k > 0}  r^{-k+1}\sqrt{\log \frac{\mu(A_{k-1}(t))}{\mu(A_k(t))}}
\end{align*}
and
\begin{align*}
	M(\mu,\nu) \le \sum_{k > 0} r^{-k+1}\sum_{B \in \cA_{k-1}} \sum_{A \in \cA_{k}(B)} \nu(A)\sqrt{\log \frac{\mu(B)}{\mu(A)}}.
\end{align*}
We can express these bounds in information-theoretic terms. Let us define, for any $k \ge 1$ and any $B \in \cA_j$ with $j < k$, the entropy and the cross-entropy
\begin{align*}
	H_{\cA_k|B}(\mu) &\deq -\sum_{A \in \cA_k}\mu(A|B)\log \mu(A|B), \\
	H_{\cA_k|B}(\mu,\nu) &\deq -\sum_{A \in \cA_k}\mu(A|B)\log \nu(A|B).
\end{align*}
Using Jensen's inequality and the bound $\inf_\mu \sup_t M(\mu,\delta_t) \le \sup_\mu M(\mu,\mu)$ \cite{Fernique_1978,Bednorz2015}, we can obtain the following:
\begin{proposition}\label{prop:it_mm_upper_bound} For any $r > 1$,
	\begin{align*}
	&	\E\Big[\sup_{t \in T}X_t\Big]\nonumber\\
	&\qquad  \lesssim \inf_{\cA}\sup_{\mu}  \sum_{k > 0} r^{-k+1}\sum_{B \in \cA_{k-1}} \mu(B) \sqrt{H_{\cA_k|B}(\mu)},
	\end{align*}
	where the infimum is over all increasing sequences $\cA = (\cA_k)_{k \ge 1}$ of partitions of $T$ with ${\rm diam}(A) \le r^{-k}$ for all $A \in \cA_k, k \ge 1$.
\end{proposition}

	\section{The upper bound and some consequences}
	
To keep things simple, we will assume for the remainder of the paper that $T$ is finite (the case of infinite $T$ is not much more difficult under the mild assumption of separability of all relevant random processes). Let $\cC = \{(\pi_k,f_k)\}_{k \in \Naturals}$ be an admissible sequence of VLC's on $T$, where each $(\pi_k,f_k)$ operates at resolution $\rho_k$. Let also $\bd{p} = (p_k)_{k \ge 1}$ be a positive probability mass function on $\Naturals$ and define the functional
\begin{align}\label{eq:barsigma}
	\bar{\sigma}_{(\cC,\bd{p})}(t) &\deq \sum_{k > 0} \rho_{k-1}\sqrt{\ln \frac{2^{\ell(f_k \circ \pi_k(t))+1}}{p_k}},
\end{align}
where we have set $\rho_0 = {\rm diam}(T) \equiv 1$. The following is a generalization of \cite[Thm.~6]{Maurer_majorizing}:

\begin{theorem}\label{thm:basic_estimate} Let $(X_t)_{t \in T}$ be a centered random process satisfying the increment condition \eqref{eq:increment}. Let $(\cC,\bd{p})$ be such that $\bar{\sigma}_{(\cC,\bd{p})}(t)$ is finite for all $t \in T$. Fix an arbitrary point $t_0 \in T$. Then, for any $u > 0$,
	\begin{align*}
		|X_t - X_{t_0}| \le \bar{\sigma}_{(\cC,\bd{p})}(t) (u+1), \qquad \forall t \in T
	\end{align*}
	with probability at least $1 - e^{-u^2}$.
\end{theorem}
\begin{proof}  Since $T$ is finite, for each $t \in T$ there exists some $k = k(t) \in \Naturals$, such that $t = \pi_k(t) = \pi_{k+1}(t) = \dots$. Therefore we can form the chaining decomposition \cite[p.~19]{Talagrand_2014}
	\begin{align*}
		X_t - X_{t_0} = \sum_{k > 0} (X_{\pi_k(t)}-X_{\pi_{k-1}(t)}).
	\end{align*}
For each $k$, let $S_k \deq \pi_k(T)$ and define  $\xi_k : S_k \to \Reals$ by
\begin{align*}
	\xi_k(s) \deq d(s,\pi_{k-1}(s)) \sqrt{\ln \frac{2^{\ell(f_k(s))+1}}{p_k} + u^2}
\end{align*}
and the event $E_k \deq \big\{ \exists s \in S_k : |X_s - X_{\pi_{k-1}(s)}| > \xi_k(s) \big\}$. Using the increment condition \eqref{eq:increment} and the Kraft--McMillan inequality, we have
\begin{align*}
	\Pr[E_k] &\le \sum_{s \in S_k} \Pr[|X_s - X_{\pi_{k-1}(s)}| > \xi_k(s)] \\
	&\le 2 \sum_{s \in S_k} \exp\Bigg(-\frac{\xi^2_k(s)}{d^2(s,\pi_{k-1}(s))}\Bigg) \\
	&= 2 \sum_{s \in S_k} \exp\Bigg(-\ln \frac{2^{\ell(f_k(s))+1}}{p_k} - u^2\Bigg) \\
	&=  p_k e^{-u^2}\sum_{s \in S_k} 2^{-\ell(f_k(s))} \\
	&\le  p_k e^{-u^2}.
\end{align*}
Then
\begin{align*}
	&\Pr \Bigg[ \exists t \in T,\, |X_t - X_{t_0}| > \sum_{k > 0}  \xi_k(\pi_k(t))\Bigg] \nonumber\\
	&\le \Pr \Bigg[\exists t \in T,\, \sum_{k > 0} \Bigg(|X_{\pi_k(t)} - X_{\pi_{k-1}(t)}| - \xi_k(\pi_k(t))\Bigg) > 0 \Bigg] \\
	&\le \Pr\Bigg[\bigcup_k E_k \Bigg] \\
	&\le e^{-u^2},
\end{align*}
where the last step follows  from the assumption on the weights $p_k$. Thus, since $d(\pi_k(t),\pi_{k-1}(t)) \le \rho_{k-1}$, 
\begin{align*}
	|X_t - X_{t_0}| \le \bar{\sigma}_{(\cC,\bd{p})}(t) + u \sum_{k > 0} \rho_{k-1}, \quad \forall t \in T
\end{align*}	
with probability at least $1-e^{-u^2}$. Since $\ell(f_k \circ \pi_k(t)) \ge 1$ for $k \ge 1$ by the unique decodability of $f_k : S_k \to \{0,1\}^*$, it follows that $\frac{2^{\ell(f_k\circ \pi_k(t))+1}}{p_k} \ge 4$, and thus
\begin{align*}
	\sum_{k > 0} \rho_{k-1} \le \bar{\sigma}_{(\cC,\bd{p})}(t)
\end{align*}
for all $t \in T$. This completes the proof.
\end{proof}

\begin{corollary} Under the same assumptions as in Theorem~\ref{thm:basic_estimate},
	\begin{align*}
		\E\Big[\sup_{t \in T}X_t\Big] \le 2\inf_{(\cC,\bd{p})}\sup_{t \in T} \bar{\sigma}_{(\cC,\bd{p})}(t).
	\end{align*}
\end{corollary}
\begin{proof} Let $\bar{a} \deq \sup_{t \in T} \bar{\sigma}_{(\cC,\bd{p})}(t)$ and $\bar{Z} \deq \sup_{t \in T}|X_t - X_{t_0}|$. By Thm.~\ref{thm:basic_estimate}, $\Pr[\bar{Z} > \bar{a} + u] \le e^{-u^2/\bar{a}^2}$ for all $u > 0$, 	so 
	\begin{align*}
		\E\Big[\sup_{t \in T}X_t \Big] &= \E\Big[\sup_{t \in T}(X_t - X_{t_0})\Big] \\
		&\le \E[\bar{Z}]\\
		&= \int^\infty_0 \Pr[\bar{Z} > r]\dif r \\
		&= \int^{\bar{a}} \Pr[\bar{Z} > r]\dif r + \int^\infty_{\bar{a}} \Pr[\bar{Z} > r] \dif r \\
		&\le \bar{a} + \int^\infty_0 e^{-u^2/\bar{a}^2} \dif u \\
		&\le 2\bar{a}.
	\end{align*}
Now take the infimum over all $(\cC,\bd{p})$.
\end{proof}

\subsection{Some consequences}
\label{ssec:consequences}

A close examination of the proof of Theorem~\ref{thm:basic_estimate} shows that we have not exploited all of the properties of admissible sequences of VLC's. To obtain further results, it will be convenient to decompose the functional $\bar{\sigma}_{(\cC,\bd{p})}$ as follows:
\begin{align*}
	&\bar{\sigma}_{(\cC,\bd{p})}(t) \le \sigma_\cC(t) + \sigma'_{(\cC,\bd{p})}(t) \\
	&\quad\deq \sum_{k > 0} \rho_{k-1} \sqrt{\ell(f_k \circ \pi_k(t))} + \sum_{k > 0} \rho_{k-1}\sqrt{\ln \frac{2}{p_k}}
\end{align*}
We can always choose $\bd{p}$ so that $\sigma'_{(\cC,\bd{p})}(t)$ can be upper-bounded by an absolute constant--e.g., with $p_k = 2^{-k}$ and $\rho_k = r^{-k} \le 2^{-k}$, we will have
\begin{align*}
	\sigma'_{(\cC,\bd{p})}(t) \le \sum_{k > 0}2^{-k+1}\sqrt{(k+1)\ln 2} \le 3.
\end{align*}
Thus, in what follows we can focus on $\sigma_\cC(t)$.

\begin{theorem}\label{thm:refinement} Let $\cC = \{(\pi_k,f_k)\}_{k \ge 1}$ be an admissible sequence of VLC's with $\rho_k = r^{-k}$ for some $r \ge 2$. Then
	\begin{align}\label{eq:sigma_bound}
		\!\!\sigma_\cC(t) \lesssim \sum_{k > 0} r^{-k+1} \sqrt{\ell(f_k \circ \pi_k(t))-\ell(f_{k-1} \circ \pi_{k-1}(t))}.
	\end{align}
\end{theorem}

\begin{proof} Let $\ell_k(t) \deq \ell(f_k \circ \pi_k(t))$, with $\ell_0(t) \equiv 0$. By admissibility of $\cC$, $(\ell_k(t))_{k \ge 1}$ is a nondecreasing sequence of positive integers. Therefore,
	\begin{align*}
		&\sigma_{\cC}(t) = \sum_{k > 0} r^{-k+1} \sqrt{\ell_k(t)} \\
		&\le \sum_{k > 0} r^{-k+1} \sqrt{\ell_k(t)-\ell_{k-1}(t)} + \sum_{k > 0} r^{-k+1} \sqrt{\ell_{k-1}(t)} \\
		&= \sum_{k > 0} r^{-k+1} \sqrt{\ell_k(t)-\ell_{k-1}(t)} + r^{-1}\sigma_\cC(t),
	\end{align*}
	which, together with $(1-r^{-1})^{-1} \le 2$, gives \eqref{eq:sigma_bound}.
\end{proof}
\noindent For each $k > 0$, the quantity under the square root in the upper bound \eqref{eq:sigma_bound} is the extra codelength needed to \textit{refine} the description $\pi_{k-1}(t)$ at resolution $r^{-k+1}$ to a more accurate one $\pi_k(t)$ at resolution $r^{-k}$. Several existing results from the literature can now be obtained as special cases.

First, we can recover the upper bound in terms of \textit{labeled nets} \cite{Guedon_2003,vanHandel_2016}. A labeled net on $(T,d)$ is a pair $(\cA,L)$, where $\cA = (\cA_k)_{k \ge 1}$ is an increasing sequence of partitions of $T$, such that ${\rm diam}(A) \le r^{-k}$ for all $A \in \cA_k$ and all $k \ge 1$, and $L : \cA \to \Naturals$ is a \textit{labeling function} satisfying the condition $\{ L(A) : A \in \cA_{k+1}(B)\} = \{1,\ldots,|\cA_{k+1}(B)|\}$ for all $B \in \cA_k$ and all $k$.

\begin{theorem}\label{thm:labeled_net} For each labeled net $(\cA,L)$, there exists an admissible sequence of VLC's, such that
	\begin{align}\label{eq:labeled_net}
		\sigma_\cC(t) \le c_1\sum_{k > 0} r^{-k+1}\sqrt{\log L(A_{k}(t))} + c_2,
	\end{align}
	where $c_1,c_2 > 0$ are universal constants.
\end{theorem}
\begin{proof} We will use the construction described in Section~\ref{sec:VLC}. To that end, define a subprobability measure
	\begin{align*}
		\mu'_1(\cdot) = \frac{6}{\pi^2}\sum_{A \in \cA_1} \frac{1}{L^2(A)|A|}\1_A(\cdot)
	\end{align*}
	so that $\mu'_1(A) = \frac{6}{\pi^2 L^2(A)}$ for each $A \in \cA_1$, and, for each $k$ and $B \in \cA_k$,  a  subprobability measure
\begin{align*}
	\nu'_{k+1}(\cdot|B) \deq \frac{6}{\pi^2}\sum_{A \in \cA_{k+1}(B)} \frac{1}{L^2(A)|A|}\1_A(\cdot),
\end{align*}
so that $\nu'_{k+1}(A|B) = \frac{6}{\pi^2L^2(A)}$ for every $A \in \cA_{k+1}(B)$. Consequently, there exist probability measures $\mu_1(\cdot) \ge \mu'_1(\cdot)$ and  $\nu_{k+1}(\cdot|B) \ge \nu'_{k+1}(\cdot|B)$. Then, with $\mu_k$ defined recursively via \eqref{eq:mu_k}, we have
	\begin{align*}
	&	\ell(f_{k+1} \circ \pi_{k+1}(t)) - \ell(f_k \circ \pi_k(t)) \nonumber \\
	&\le \log \frac{1}{\mu_{k+1}(A_{k+1}(t))} - \log \frac{1}{\mu_k(A_k(t))} + 1 \\
	&= \log \frac{1}{\nu_{k+1}(A_{k+1}(t))|A_k(t))} + 1 \\
	&\le \log \frac{1}{\nu'_{k+1}(A_{k+1}(t)|A_k(t))} + 1 \\
	&\le 2 \log L(A_{k+1}(t)) + \log \frac{\pi^2}{6} + 1
	\end{align*}
The inequality \eqref{eq:labeled_net} then follows from Theorem~\ref{thm:refinement}.
\end{proof}

\noindent In a similar way, we can recover a bound of Bednorz \cite{Bednorz2015} that involves an increasing partition just like above:
\begin{theorem}\label{thm:Bednorz_upper} For each sequence of partitions as in Theorem~\ref{thm:labeled_net} and each $\mu \in \cP(T)$, there exists an admissible sequence $\cC$ of VLC's, such that
	\begin{align}\label{eq:Bednorz}
		\sigma_\cC(t) \le \sum_{k > 0} r^{-k+1} \sqrt{\log \frac{\mu(A_{k-1}(t))}{\mu(A_{k}(t))}} + 2
	\end{align}
\end{theorem}
\begin{proof} Without loss of generality, we may assume that $\mu(A) > 0$ for all $A \in \cA_k$ and all $k$. We follow the same steps as in the proof of Theorem~\ref{thm:refinement}, but take $\mu_1 = \mu$ and $\nu_{k+1}(\cdot|B) = \mu(\cdot|B)$ for all $k$ and all $B \in \cA_k$. Then, for every $B \in \cA_k$ and $A \in \cA_{k+1}(B)$, $\nu_{k+1}(A|B) = \mu(A|B) = \frac{\mu(A)}{\mu(B)}$, and 
	\begin{align*} 
		&\ell(f_{k+1} \circ \pi_{k+1}(t)) - \ell(f_k \circ \pi_k(t)) \nonumber\\
		&\le \log \frac{\mu(A_k(t))}{\mu(A_{k+1}(t))} + 1 \\
		&= \log \frac{\mu(A_k(t))}{\mu(A_{k+1}(t))} + 1.
	\end{align*}
Summing over $k$, we get \eqref{eq:Bednorz}.
\end{proof}

\noindent In fact, there is no loss of generality in working with admissible sequences of VLC's constructed  from a single $\mu \in \cP(T)$ rather than from a recursively defined sequence $(\mu_k)_{k \ge 1}$:

\begin{theorem} For any $\bd{p} = (p_k)_{k \ge 1}$ and any admissible sequence $\cC = \{(\pi_k,f_k)\}_{k \ge 1}$ of VLC's constructed from an increasing sequence of partitions $\cA = (\cA_k)_{k \ge 0}$ and a sequence of probability measures $(\mu_k)_{k \ge 1}$ via \eqref{eq:mu_k}, there exists a probability measure $\mu \in \cP(T)$ and an admissible sequence $\cC' = \{(\pi'_k,f'_k)\}_{k \ge 1}$ of VLC's, such that
	\begin{align}\label{eq:VLC_comparison}
		\bar{\sigma}_{(\cC',\bd{p})}(t) \le \bar{\sigma}_{(\cC,\bd{p})}(t) + \sum_{k > 0} r^{-k+1}\sqrt{\log \frac{2}{p_k}}, \quad t \in T.
	\end{align}
\end{theorem}
\begin{proof} Let $\cC = \{(\pi_k,f_k)\}_{k \ge 1}$ be constructed from $\cA$ and $(\mu_k)_{k \ge 1}$ as indicated. Given $\bd{p}$, form the mixture $\mu = \sum_{k > 0}p_k \mu_k$. Then, keeping the same sequence of partitions $\cA = (\cA_k)_{k \ge 1}$, construct $\cC' = \{(\pi'_k,f'_k)\}_{k \ge 1}$ with $\pi'_k = \pi_k$ and $f'_k$ satisfying $\ell(f'_k \circ \pi'_k (t)) = \lceil \log \frac{1}{\mu(A_k(t))}\rceil$ for every $t \in T$ and every $k \in T$. Then
	\begin{align*}
		\ell(f'_k \circ \pi'_k(t)) &\le \log \frac{1}{\mu(A_k(t))} + 1 \\
		&\le \log \frac{1}{\mu_k(A_k(t))} + \log \frac{2}{p_k} \\
		&\le \ell(f_k \circ \pi_k(t)) + \log \frac{2}{p_k}.
	\end{align*}
	Using this in \eqref{eq:barsigma} with $\rho_{k} = r^{-k}$, we get \eqref{eq:VLC_comparison}.
\end{proof}
\noindent Conversely, any choice of an admissible sequence $\cC = \{(\pi_k,f_k)\}_{k \ge 1}$ of VLC's and $\bd{p} = (p_k)_{k \ge 1}$ gives rise to a probability measure
\begin{align*}
	\mu \propto \sum_{k > 0} p_k \sum_{s \in \pi_k(T)} 2^{-\ell(f_k(s))}\delta_s
\end{align*}
(see, e.g., \cite[Sec.~4]{Maurer_majorizing}), and it is readily verified that, for each $t \in T$ and each $k$,
\begin{align*}
	\log \frac{1}{\mu(A_k(t))} \le \ell(f_k \circ \pi_k(t)) + \log \frac{1}{p_k}.
\end{align*}
Consequently, the machinery developed in this section can be used to give a self-contained proof of Prop.~\ref{prop:it_mm_upper_bound}. In fact, using a construction due to Maurer \cite{Maurer_majorizing}, we can completely eliminate the explicit mminimization over increasing sequences of partitions, as in Prop.~\ref{prop:it_mm_upper_bound}--instead, given the choice of $\mu \in \cP(T)$ (or, equivalently, a uniquely decodable code $f : T \to \{0,1\}^*$), the mappings $\pi_k$ and $f_k$ at resolution level $r^{-k}$ are chosen so that, for each $t \in T$, the binary string $f_k \circ \pi_k(t)$ is obtained by truncating the lossless representation $f(t)$ to the smallest number of bits, such the diameter of the set of all $t' \in T$ that could be `confused' with $t$ on the basis of this truncated representation is at most $r^{-k}$. This set is then taken as $A_k(t)$. We leave the details to the full version of this work.

\section{The lower bound for Gaussian processes}

We now consider the case when the process $(X_t)_{t \in T}$ is Gaussian and $T$ is endowed with the canonical metric \eqref{eq:can_metric}. We will use \(L,L_1,\ldots\) to denote various absolute constants. Note that their
values are not necessarily the same in different appearances.
The main technical tool for the lower bound is the following improvement of Sudakov minoration
\cite[Prop.~2.4.9]{Talagrand_2014}:
\begin{proposition}[]\label{prop:sudakov} Let $t_1,\dots,t_m \in T$ be such that \(d\left(t_{\ell}, t_{\ell^{\prime}}\right) \geq a\) if \(\ell \neq \ell^{\prime}\). Consider \(b > 0\) and, for each $\ell \le m$,  a finite set \(H_{\ell} \subset B\left(t_{\ell}, b\right)\). Then for \(H=\bigcup_{\ell \leq m} H_{\ell}\) we have
\begin{align*}
\E \Big[\sup _{t \in H} X_t\Big] \geq \frac{a}{L_1} \sqrt{\log m}-L_2 b \sqrt{\log m}+\min _{\ell \leq m} \E \Big[\sup _{t \in H_{\ell}} X_t\Big].
\end{align*}
When \(b \leq a /\left(2 L_1 L_2\right)\), we have
\begin{align*}
\E \Big[\sup _{t \in H} X_t\Big] \geq \frac{a}{2 L_1} \sqrt{\log m}+\min _{\ell \leq m} \E \Big[\sup _{t \in H_{\ell}} X_t\Big].
\end{align*}
\end{proposition}

\noindent Another ingredient is the construction of an increasing sequence of partitions
on \(T\).
We follow Bednorz \cite{Bednorz2015} and Talagrand \cite{Talagrand_1992} and
use a greedy procedure to construct the
partitions $(\cA_k)_{k \ge 0}$ for a fixed $r > 1$.
For any $A \subseteq T$, let $G(A) \deq \E[\sup_{t \in A}X_t]$. Let $\cA_0 =\{T\}$. To construct $\cA_k$ for $k \ge 1$ from $\cA_{k-1}$, we partition each $B \in \cA_{k-1}$ into sets $A_1,\ldots,A_m$ as follows: Let $B_0 \deq B$ and choose $t_1 \in B_0$ so that
\begin{align*}
	G(B(t_1,r^{-k-1}/2)\cap B_0) = \max_{t \in B_0} G(B(t,r^{-k-1}/2)\cap B_0).
\end{align*}
Let $A_1 \deq B(t_1,r^{-k}/2) \cap B_0$ and $B_1 \deq B_0 \setminus A_1$. Now iterate this: If, for $i \ge 1$, $B_{i-1} \neq \varnothing$, choose $t_i \in B_{i-1}$ so that
\begin{align*}
	& G(B(t_i,r^{-k-1}/2)\cap B_{i-1}) \nonumber\\
	&\quad = \max_{t \in B_{i-1}} G(B(t,r^{-k-1}/2)\cap B_{i-1}). 
\end{align*}
Let $A_i \deq B(t_i,r^{-k}/2) \cap B_{i-1}$ and $B_i \deq B_{i-1} \setminus A_i$ and continue. This procedure is guaranteed to terminate for some $m \le N(T,d,r^{-k}/2)$, where $N(T,d,\cdot)$ are the covering numbers of $(T,d)$.

We are now ready to state the majorizing measure theorem via codes.
\begin{theorem} \label{thm:unif_lower} Let $(X_t)_{t \in T}$ be a centered Gaussian process. For any $r \ge \max\{2,L\}$, there exists an admissible sequence $\cC$ of VLC's with $\rho_k \le r^{-k}$, such that
    \begin{align}
        \E\Big[\sup_{t \in T} X_t\Big] \gtrsim \frac{1}{r} \sigma_{\cC}(t), \qquad \forall \, t \in T.
    \end{align}
\end{theorem}
\begin{proof}
Let $L = L_1 L_2$, where $L_1,L_2$ are the contants from Prop.~\ref{prop:sudakov}. Given $r \ge L$, let $(\cA)_{k\ge 0}$ be an increasing sequence of partitions as described above.
  For each $k$ and each $A \in \cA_k$, pick a unique point $s_A \in A$ and then let $\pi_k(t) \deq s_{A_k(t)}$, where $A_k(t)$ is the unique element of $\cA_k$ containing $t$.

We construct the codes $f_k$ inductively. Set $f_0(\pi_0(t))={\mathsf e}$, the empty string in $\{0,1\}^*$. Assuming $f_0, \dots, f_{k-1}$ have been constructed, we assign code lengths $\ell_k(s)$ to the points $s \in \pi_k(T)$ and show that this assignment satisfies the Kraft--McMillan inequality. Thus, there exists a prefix code $f_k : \pi_k(T) \to \{0,1\}^*$ with $\ell(f_k(s)) = \ell_k(s)$ for all $s \in \pi_k(T)$.

For each $s \in \pi_{k}(T)$, let $i_k(s) \in \Naturals$ be the index
of the set $A_k(s) \equiv \pi_{k}^{-1}(s)$ as a partition element in $A_{k-1}(s) \equiv \pi_{k-1}^{-1}(s)$ inherited from the above greedy construction procedure. Then, with
$\ell_k(s) = \ell(f_{k-1}\circ \pi_{k-1}(s)) + \log(i_k(s)+1)^2$, we have
\begin{align*}
&\sum_{s \in \pi_k(T)}2^{-\ell_k(s)} \\
&= \sum_{t \in \pi_{k-1}(T)} \sum_{s:\pi_{k-1}(s)=t}
            2^{-\ell(f_{k-1}(s))-\log(i_k(s)+1))^2}\\
    &\le\sum_{t \in \pi_{k-1}(T)} 2^{-\ell(f_{k-1}(t))}\sum_{i\ge 1} \frac{1}{(i+1)^2} \le 1.
\end{align*}
Hence the existence of a uniquely decodable (prefix) code $f_k$ is guaranteed by
Kraft--McMillan inequality.  Then, by the construction procedure and Prop.~\ref{prop:sudakov},
for each
$t \in T$,
\begin{align*}
    G( A_{k-1}(t)) \ge G(A_{k+1}(t)) + L r^{-k}\sqrt{\log(i_k(\pi_k(t)))}.
\end{align*}
Since $1+\sqrt{\log i} \ge \sqrt{1+\log i }\ge \sqrt{\log(i+1)}$ for $i \ge 1$, 
\begin{align*}
    &G( A_{k-1}(t)) + L r^{-k} \\
    &\ge G(A_{k+1}(t)) + L r^{-k}\sqrt{\log(i_k(\pi_k(t))+1)^2}\\
    &=G(A_{k+1}(t)) + L r^{-k} \sqrt{\ell_k(t)-\ell_{k-1}(t)},
\end{align*}
where $\ell_k(t)=\ell(f_k\circ \pi_k(t))$. Applying induction on $G(A_0(t))$ and $G(A_1(t))$,
we obtain
\begin{align}\label{eq:len_diff}
    G(T) + L \gtrsim \sum_k  r^{-k} \sqrt{\ell_k(t)-\ell_{k-1}(t)}, \qquad \forall \, t\in T
\end{align}
Thm.~\ref{thm:refinement} and the fact that $G(T) \gtrsim {\rm diam}(T)$, so that the universal constant can be absorbed by $G(T)$, give the result.
\end{proof}

\section*{Acknowledgments}

The authors would like to thank Daniel Roy for pointing out some typos and inaccuracies in the proof of Theorem~\ref{thm:basic_estimate}.

\bibliography{majorizing_codes.bbl}

\begin{thebibliography}{10}
\providecommand{\url}[1]{#1}
\csname url@samestyle\endcsname
\providecommand{\newblock}{\relax}
\providecommand{\bibinfo}[2]{#2}
\providecommand{\BIBentrySTDinterwordspacing}{\spaceskip=0pt\relax}
\providecommand{\BIBentryALTinterwordstretchfactor}{4}
\providecommand{\BIBentryALTinterwordspacing}{\spaceskip=\fontdimen2\font plus
\BIBentryALTinterwordstretchfactor\fontdimen3\font minus
  \fontdimen4\font\relax}
\providecommand{\BIBforeignlanguage}[2]{{%
\expandafter\ifx\csname l@#1\endcsname\relax
\typeout{** WARNING: IEEEtran.bst: No hyphenation pattern has been}%
\typeout{** loaded for the language `#1'. Using the pattern for}%
\typeout{** the default language instead.}%
\else
\language=\csname l@#1\endcsname
\fi
#2}}
\providecommand{\BIBdecl}{\relax}
\BIBdecl

\bibitem{Talagrand_2014}
M.~Talagrand, \emph{Upper and Lower Bounds for Stochastic Processes}.\hskip 1em
  plus 0.5em minus 0.4em\relax Springer, 2014.

\bibitem{Fernique_1975}
X.~Fernique, ``Regularit\'e des trajectoires des fonctions al\'eatoires
  gaussiennes,'' in \emph{Ecole d'Et{\'{e}} de Probabilit{\'{e}}s de
  Saint-Flour {IV}{\textemdash}1974}.\hskip 1em plus 0.5em minus 0.4em\relax
  Springer, 1975, pp. 1--96.

\bibitem{Talagrand_1987}
M.~Talagrand, ``Regularity of gaussian processes,'' \emph{Acta Mathematica},
  vol. 159, pp. 99--149, 1987.

\bibitem{Guedon_2003}
O.~Gu{\'{e}}don and A.~Zvavitch, ``Supremum of a process in terms of trees,''
  in \emph{Geometric Aspects of Functional Analysis}.\hskip 1em plus 0.5em
  minus 0.4em\relax Springer, 2003, pp. 136--147.

\bibitem{Maurer_majorizing}
A.~Maurer, ``Majorizing codes and measures,'' preprint, available at
  http://www.andreas-maurer.eu/mmnotes.pdf, 2010.

\bibitem{Cover_Thomas}
T.~M. Cover and J.~A. Thomas, \emph{Elements of Information Theory},
  2nd~ed.\hskip 1em plus 0.5em minus 0.4em\relax Wiley, 2006.

\bibitem{Kontoyiannis_2002}
I.~Kontoyiannis and J.~Zhang, ``Arbitrary source models and {B}ayesian
  codebooks in rate-distortion theory,'' \emph{IEEE Trans. Inform. Theory},
  vol.~48, no.~8, pp. 2276--2290, August 2002.

\bibitem{Massart_concentration}
P.~Massart, \emph{Concentration Inequalities and Model Selection}.\hskip 1em
  plus 0.5em minus 0.4em\relax Springer, 2007.

\bibitem{Audibert_Bousquet}
J.-Y. Audibert and O.~Bousquet, ``Combining {PAC-Bayesian} and generic chaining
  bounds,'' \emph{Journal of Machine Learning Research}, vol.~8, no.~32, pp.
  863--889, 2007.

\bibitem{Bednorz2015}
W.~Bednorz, ``The majorizing measure approach to sample boundedness,''
  \emph{Colloquium Mathematicae}, vol. 139, no.~2, pp. 205--227, 2015.

\bibitem{Fernique_1978}
X.~Fernique, ``Caract\'erisation de processus \`a trajectoires major\'ees ou
  continues,'' \emph{S\'eminaire de probabilit\'es de Strasbourg}, vol.~12, pp.
  691--706, 1978.

\bibitem{vanHandel_2016}
R.~{van Handel}, ``Probability in high dimension,'' Princeton University
  lecture notes, 2016.

\bibitem{Talagrand_1992}
M.~Talagrand, ``A simple proof of the majorizing measure theorem,''
  \emph{Geometric and Functional Analysis}, vol.~2, no.~1, pp. 118--125, March
  1992.

\end{thebibliography}

\end{document}